\documentclass[a4paper,reqno]{amsart}
\usepackage{amsmath, amssymb, eucal, amscd, amstext, verbatim, enumerate, color, comment}
\usepackage[pagewise]{lineno}

\newtheorem{thm}{Theorem}[section]
\newtheorem{lem}[thm]{Lemma}
\newtheorem{eg}[thm]{Example}
\newtheorem{prop}[thm]{Proposition}
\newtheorem{cor}[thm]{Corollary}
\newtheorem{defn}[thm]{Definition}
\newtheorem{rem}[thm]{Remark}

\newtheorem{conj}[thm]{Conjecture}

\newtheorem{thm2}{Theorem}
\newtheorem{prop2}[thm2]{Proposition}

\newcommand{\smnoind}{\smallskip\noindent}

\newcommand{\sa}{{\rm sa}}

\newcommand{\ti}{\tilde}

\newcommand{\mo}{\mathrm{o}}
\newcommand{\mc}{\mathrm{c}}
\newcommand{\mm}{\mathbf{m}}
\newcommand{\mq}{\mathrm{q}}

\newcommand{\RP}{\mathbb{R}_+}

\newcommand{\BC}{\mathbb{C}}
\newcommand{\BZ}{\mathbb{Z}}
\newcommand{\BN}{\mathbb{N}}
\newcommand{\BR}{\mathbb{R}}
\newcommand{\BM}{\mathbb{M}}
\newcommand{\BP}{\mathbb{P}}

\newcommand{\CS}{\mathcal{S}}

\newcommand{\CC}{\mathcal{C}}

\newcommand{\CP}{\mathcal{P}}
\newcommand{\CB}{\mathcal{B}}
\newcommand{\CO}{\mathcal{O}}
\newcommand{\CQ}{\mathcal{Q}}
\newcommand{\CR}{\mathcal{R}}
\newcommand{\CL}{\mathcal{L}}
\newcommand{\CT}{\mathcal{T}}

\newcommand{\CU}{\mathcal{U}}
\newcommand{\CM}{\mathcal{M}}

\newcommand{\KI}{\mathfrak{I}}

\newcommand{\KG}{\mathfrak{G}}

\newcommand{\KH}{\mathfrak{H}}
\newcommand{\KP}{\mathfrak{P}}

\newcommand{\bs}{\mathbf{s}}
\newcommand{\bz}{\mathbf{z}}

\newcommand{\ba}{\mathbf{a}}

\newcommand{\hull}{\mathrm{hull}}

\begin{document}
\title{Quantum sets and Gelfand spectra\\
\emph{(names for some terminologies are changed in the published version and the new title is ``Ortho-sets and Gelfand spectra'')}}

\author{Chun Ding \and Chi-Keung Ng}

\address[Chun Ding]{Mathematical Institute, Leiden University, The Netherlands.}
\email{c.ding@math.leidenuniv.nl}

\address[Chi-Keung Ng]{Chern Institute of Mathematics and LPMC, Nankai University, Tianjin 300071, China.}
\email{ckng@nankai.edu.cn; ckngmath@hotmail.com}



\date{\today}
\maketitle

\begin{abstract}
Motivated by quantum states with zero transition probability, we introduce the notion of quantum set which is a set equipped with a relation $\neq_\mathrm{q}$ satisfying: $x\neq_\mathrm{q} y$ implies both $x\neq y$ and $y \neq_\mathrm{q} x$. 
For a quantum set, a canonical complete ortholattice is constructed. 
This ortholattice is orthomodular if and only if the quantum set satisfies a canonical condition concerning subsets. 
This produces a surjective correspondence from the collection of quantum sets to the collection of complete ortholattices, and hence, the theory of quantum sets captures almost everything about quantum logic.

We also introduce the Gelfand spectrum for a quantum system modeled on the self-adjoint part $B_\mathrm{sa}$ of a $C^*$-algebra $B$, by defining a ``quantum topology'' on the quantum set  
of pure states of $B$, via a hull-kernel construction using closed left ideals.  
We establish a generalization of the Gelfand theorem by showing that a bijection between the ``semi-classical object'' of the Gelfand spectra of two quantum systems that preserves the respective quantum topological structures is induced by a Jordan isomorphism between the self-adjoint parts (i.e. an isomorphism of the quantum systems), when the underlying $C^*$-algebras satisfy a mild condition. 
\end{abstract}

\medskip

\section{Introduction}

\medskip

In the classical world, two points are either the same or distinct; but in the quantum world, one sometimes wants to consider a more restrictive form of ``distinctness''; e.g., two quantum states are thought to be ``really distinct''  if they have zero transition probability. 
We give a mathematical presentation of this by introducing the notion of ``quantum sets''. 

\medskip

Let us state these in more precise terms. 
We call a symmetric subrelation $\neq_\mq$ of the usual distinctness relation $\neq$ on a set $X$ a ``q-distinctness relation''. 
In this case, we say that $(X,\neq_\mq)$ is a quantum set. 
The q-distinctness relation induces a q-complement on $X$ (see \eqref{eqt:def-quantum compl}). 
Subsets of $X$ are called q-subsets if they are the q-complements of some other subsets.
The collection of all q-subsets of $X$ will be denoted by $\CQ(X)$. 

\medskip

Note that quantum sets as above are different from the varies notions of quantum set theory found in the literature; e.g., \cite{Gud,Schles,Takeuti,TK}.
They are also different from quantum sets as defined in \cite[Definition 9]{DRZ}. 

\medskip

One may identify a quantum set in our sense as a graph with no loop (where two elements in the set are joined by an edge if they are q-distinct). 
In this respect, ``quantum automorphism groups'' of finite quantum sets have already been studied in the literature (see, e.g., \cite{Ban, BBC, BGS, JSW, MRV, Sch}).

\medskip

Several natural examples of quantum sets will be given in Section 2. 
In particular, if certain system is modeled on a metric space such that there is an ``uncertainty'', in the sense that two points cannot be theoretically separated when their distance is less than a certain amount $\delta$, then one can view the system as a quantum set. 
In this case, the  ``logic'' of this system (i.e. the theory of q-subsets) will be very different from the classical ones. 
We will have a closer look at the case of the Euclidean metric space $\BR$ in Example \ref{eg:uncert}. 
An example concerning transition probability will also be considered (see Example \ref{eg:tran-prob}). 

\medskip

We say that a quantum set is atomic if all the singleton subsets are q-subsets. 
Moreover, we say that a quantum set $X$ is hereditary if for every element $T\in \CQ(X)$, one has 
$$\CQ(T) = \{S\in \CQ(X): S\subseteq T \},$$ 
when $T$ is equipped with the q-distinctness relation induced from $X$. 

\medskip

The following result  (which can be found in Lemma \ref{lem:lattice} and Proposition \ref{prop:atom-hered-q-distinct}) means that, similar to the relation between ordinary set theory and ordinary logic, quantum set theory can be regarded as the ``set theory''  behind quantum logic.

\medskip

\begin{prop2}\label{prop:quantum set-ortho-lat}
(a) If $(X, \neq_\mq)$ is a quantum set, then $\CQ(X)$ is a complete ortholattice, under the intersection, the q-union as defined in \eqref{eqt:def-q-union} and the q-complement.
Moreover, $(X,\neq_\mq)$ is hereditary if and only if $\CQ(X)$ is an orthomodular  lattice

\smnoind
(b) For an ortholattice $\CL$, there is a quantum set $\CL^\star$ with $\CL$ being a sub-ortholattice of $\CQ(\CL^\star)$. 
The assignment $\CL \mapsto \CL^\star$ is an injective correspondence from the collection of ortholattices to that of quantum sets. 
When $\CL$ is complete,  one has $\CL = \CQ(\CL^\star)$.

\smnoind
(c) There is a canonical bijective correspondence between the collection of complete atomistic ortholattices and that of atomic quantum sets (respectively, the collection of complete atomic orthomodular lattices and that of atomic hereditary quantum sets). 
\end{prop2}

\medskip

For a quantum system modeled on the self-adjoint part $B_\sa$ of a (complex) $C^*$-algebra $B$, we introduce the semi-classical object of Gelfand spectrum for this system as follows. 
Let $\KP^B$ be the set of all pure states on $B$. 
For any $\phi,\psi \in \KP^B$, we denote $\phi\neq_\mo \psi$ if $\phi$ and $\psi$ have orthogonal support projections; i.e., $\phi$ and $\psi$ has zero transition probability.  
For any left closed ideal $L\subseteq B$, we set 
$$\hull(L):=\{\phi\in \KP^B: \phi(x^*x) = 0, \text{ for every }x\in L \}.$$
Then $\hull(L)$ is a q-subset of $(\KP^B, \neq_\mo)$, and the collection of all such q-subsets form a quantum topology $\CC^B$ on $(\KP^B, \neq_\mo)$, in the sense of Definition \ref{defn:quantum top} (see Proposition \ref{prop:quantum top}). 

\medskip

Notice that our notion of 
Gelfand spectra is similar to the ideas in \cite{Ake71}. 
However, unlike \cite{Ake71}, where the atomic part of the bidual of  $B$ together with the set of q-open projections are considered, %
the Gelfand spectrum is a ``semi-classical image'' of this structure on the pure state space of $B$ (see Lemma \ref{lem:PS-min-proj}). 
On the other hand, a similar structure on the pure state space of $B$ was introduced in \cite{GK}, but the q-distinctness relation $\neq_\mo$ was not considered in \cite{GK}.

\medskip

The  Gelfand spectrum for a $C^*$-algebra captures the self-adjoint part of original algebra up to a Jordan isomorphism (under a mild assumption), which is good enough for the consideration of physical structure modeled on the self-adjoint parts of $C^*$-algebras. 
Let us recall that a linear map $\Gamma$ from the self-adjoint part of a $C^*$-algebra $A$ to that of another $C^*$-algebra is a \emph{Jordan isomorphism} if it preserves the Jordan product; i.e. $\Theta(ab + ba)= \Gamma(a)\Gamma(b) + \Gamma(b)\Gamma(a)$ ($a,b\in A_\sa$). 
The following non-commutative generalization of the Gelfand theorem will be proved in Theorem \ref{thm:main2} and Corollary \ref{cor:main2}.

\medskip

\begin{thm2}\label{thm:quantum spec}
Let $A$ and $B$ be two $C^*$-algebras.
Suppose that there is a bijection  $\Psi:\KP^A \to \KP^B$ preserving the q-distinctness relations such that $\CC^B = \big\{\Psi(C): C\in \CC^A \big\}$.

\smnoind
(a) If $A$ has no 2-dimensional irreducible $^*$-representation, then there is a Jordan isomorphism $\Gamma: B_\sa\to A_\sa$ such that $\Psi(\omega) = \omega\circ \Gamma$ ($\omega\in \KP^A$). 

\smnoind
(b) 
%
If $A$ is simple (including the case when $A=\BM_2$), then $A$ and $B$ are either $^*$-isomorphic or $^*$-anti-isomorphic. 
\end{thm2}



\medskip

The above tells us that the Gelfand spectra give a faithful presentation of quantum systems that are modelled on self-adjoint parts of $C^*$-algebras; in the sense that ``quantum homeomorphisms'' of Gelfand spectra are induced by isomorphisms of the quantum systems (if the underlying $C^*$-algebras have no $2$-dimensional irreducible representations). 
This semi-classic picture of quantum systems  may help to establish a link with classical systems.

\medskip

	If one wants a stronger conclusion of a $^*$-isomorphism in Theorem \ref{thm:quantum spec}(a), one needs to add some extra structures. 
	One possibility is the ``orientation structure'' as introduced in \cite{AHS} and \cite{Shu82}. 
	However, this structure seems a bit complicated.  
In a later work (\cite{Ng-Cat-QS}), we will introduce the notion of signature on the pure state space, and show that if the map $\Psi$ in Theorem \ref{thm:quantum spec} also preserves the canonical signatures on the pure state spaces, then $\Gamma$ can be extended to a $^*$-isomorphism from $A$ onto $B$.

\medskip

In \cite{Ng-Cat-QS}, we will also consider 
a ``larger quantum spectrum'' for a $C^*$-algebra, and show that the category of self-adjoint parts of unital $C^*$-algebras having no 2-dimensional irreducible representation, equipped with unital Jordan homomorphisms as morphisms, is a full subcategory of the category of quantum topological spaces (under as a suitable definition of morphisms). 
The functor sending $A_\sa$ to its Gelfand spectrum will also be further investigated in \cite{Ng-Cat-QS}.

\medskip

The proof of Theorem \ref{thm:quantum spec} requires results concerning q-closed projections as studied in \cite{Akemann68, Ake69, Ake71, APT73}, and the Dye theorem for von Neumann algebras (\cite{Dye}). 
More precisely, we will first prove the following version of Dye theorem for $C^*$-algebras (see Theorem \ref{thm:main}).  

\medskip

\begin{thm2}\label{thm:closed-proj}
	Let $A$ and $B$ be $C^*$-algebras such that $\BM_2$ is not a quotient $C^*$-algebra of $A$.
	If there is a bijection $\Phi$ from the set $\CC(A)$ of q-closed projections of $A$ onto that of $B$ that respects the q-distinctness relations induced by orthogonality, there is a unique Jordan isomorphism $\Gamma: A_\sa\to B_\sa$ with $\Phi = \Gamma^{**}|_{\CC(A)}$. 
\end{thm2}

\medskip

The readers should note the difference between Theorem \ref{thm:closed-proj} and \cite[Corollary I.2]{Ake71}.
In \cite[Corollary I.2]{Ake71}, one starts with a $^*$-isomorphism between the atomic parts of the biduals of the $C^*$-algebras that preserves the corresponding q-closed projections. 
However, in Theorem \ref{thm:closed-proj}, we only need isomorphism for the ``graphs of q-closed projections'' (where two projections are joined by an edge if they are orthogonal).

\medskip

Finally, let us declare that all vector spaces and algebras in this paper are over the complex field, except for the self-adjoint parts of $C^*$-algebras.

\medskip

\section{Quantum sets and their relation to ortholattices}

\medskip

When a set  $X$ is equipped with a symmetric subrelation $\neq_\mq$ of the usual distinctness relation $\neq$, we say that $(X,\neq_\mq)$ is a \emph{quantum set} and $\neq_\mq$ is called a \emph{q-distinctness relation}. 
Needless to say, $\neq$ is itself a q-distinctness relation.
In the case when $\neq_\mq$ coincides with $\neq$, we say that $(X, \neq_\mq)$ is \emph{classical}. 

\medskip

We define the \emph{q-complement} on the collection $\CP(X)$ of all subsets of $X$ as follows:
\begin{equation}\label{eqt:def-quantum compl}
D^\mc := \{y\in X: y\neq_\mq z,  \text{ for any }z\in D\} \qquad (D\in \CP(X)).
\end{equation}
Moreover, we set $D^{\mc\mc} := (D^\mc)^\mc$. 
The q-distinctness relation $\neq_\mq$ on $X$ induces a q-distinctness relation, again denoted by $\neq_\mq$, on $D$. 

\medskip

\begin{defn}
(a) If $S\in  \CP(X)$ satisfying $S = S^{\mc\mc}$, then $(S, \neq_\mq)$ is called a \emph{q-subset} of $(X, \neq_\mq)$. 
The collection of all q-subsets of $X$ will be denoted by $\CQ(X, \neq_\mq)$, or simply by $\CQ(X)$.

\smnoind
(b) If $(Y, \neq_\mq)$ is another quantum set, then a bijection $\Psi:X\to Y$ is called a \emph{strict quantum bijection} if $\Psi$ preserves the q-distinctness relation in both directions, i.e., 
$\Psi(\{x\}^\mc) = \{\Psi (x)\}^\mc$ $(x\in X).$
\end{defn}

\medskip

Note that $D\mapsto D^{\mc\mc}$ is a closure operator on $\CP(X)$ that sends $D\in \CP(X)$ to the smallest element in $\CQ(X)$ containing $D$. 

\medskip

Suppose that $C\in \CP(X)$. 
A subset $D\subseteq C$ belongs to $\CQ(C)$, when $C$ is equipped with the induced q-distinctness relation, if and only if 
\begin{equation}\label{eqt:quan-comp-in-quan-subset}
D=(D^\mc \cap C)^\mc \cap C.
\end{equation}
Note that $\CQ(C)\subseteq \CQ(X)$ if and only if $C\in \CQ(X)$. 

\medskip

It is natural to ask when we have $\CQ(T) = \{S\in \CQ(X): S\subseteq T \}$, for every q-subset  $T\in \CQ(X)$. 
Notice that it is the case if and only if the closure operator on $\CP(T)$ given by the induced q-distinctness relation on $T$ coincides with the restriction of the closure operator on $\CP(X)$ to $\CP(T)$. 

\medskip

On the other hand, it is natural to ask when every singleton subset of $X$ is a q-subset.

\medskip

\begin{defn}\label{defn:atomic-hered}
A q-distinctness relation $\neq_\mq$ on $X$ is said to be 

\begin{enumerate}[i)]
	\item \emph{atomic} if $\{x\}^{\mc\mc} = \{x\}$ for every $x\in X$ (i.e. all singleton sets are q-subsets);
	
	\item \emph{hereditary} if for each $T\in \CQ(X)$, one has $\{S\in \CQ(X): S\subseteq T \} = \CQ(T)$.
	\end{enumerate}
We say that the quantum set $(X, \neq_\mq)$ is \emph{atomic} (respectively, \emph{hereditary}) if $\neq_\mq$ is atomic (respectively, hereditary). 
\end{defn}

\medskip

\begin{defn}
Let $(\CL, \wedge, \vee)$ be a lattice with a smallest element $0$.

\smnoind
(a)  Let $\CL^\mm$ be the set of minimal elements in $\CL\setminus \{0\}$, called the \emph{atoms} of $\CL$. Then $\CL$ is said to be 
\begin{itemize}
	\item \emph{atomic} if for each $p\in \CL\setminus \{0\}$, one can find $e\in \CL^\mm$ with $e\leq p$;
	
	\item \emph{atomistic} if 
	$\bigvee \{e\in \CL^\mm: e\leq p \} \text{ exists and equals } p$, for every $p\in \CL\setminus \{0\}$. 
\end{itemize}

\smnoind
(b) Suppose that there is an operator $':\CL \to \CL$ satisfying $p = (p')'$, $0 = p \wedge p'$, and $p'\leq q'$  for every $p, q\in \CL$ with $q\leq p$, called an \emph{orthocomplementation}. 
Then $\CL$ is called an \emph{ortholattice}.
In this case, $0'$ will be denoted by $1$. 

\smnoind
(c) An ortholattice $\CL$ is called an \emph{orthomodular lattice} if for every $p,q\in \CL$ with $p\leq q$, one has $q = p \vee (q\wedge p')$ (this is called the \emph{orthomodular law}).
\end{defn}

\medskip

Some people define ``atomistic lattice'' in a way that every non-zero element is a finite join of atoms, but this definition is different from the above (unless every non-zero element dominates a finite number of atoms). 
Furthermore, as noted in \cite[p.140]{Kal83}, 
\begin{quotation}
a complete orthomodular lattice is atomistic if and only if it is atomic. 
\end{quotation}
More information on ortholattices and orthomodular lattices can be found in standard textbooks (e.g., \cite{Kal83} and \cite{Kal98}).

\medskip

Since $D\mapsto D^{\mc\mc}$ is a closure operator on $\CP(X)$, it is well-known that $\CQ(X)$ is a complete lattice under the usual conjunction $\cap$ and the adapted disjunction $\vee$, called the \emph{q-union}, defined as follows (see \cite[p.254]{Kal83}): 
\begin{equation}\label{eqt:def-q-union}
S\vee T:= (S\cup T)^{\mc\mc} \qquad (S,T\in \CQ(X)). 
\end{equation}

\medskip

\begin{lem}\label{lem:lattice}
Let $(X,\neq_\mq)$ and $(Y, \neq_\mq)$ be two quantum sets.

\smnoind
(a) $(\CQ(X), \cap, \vee,\ \!^\mc )$ is a complete ortholattice. 

\smnoind
(b) The quantum set $(X, \neq_\mq)$ is hereditary if and only if for every $S, T\in \CQ(X)$ satisfying $S\subseteq T$ and $T\cap S^\mc = \emptyset$, one has $S = T$; which is equivalent to 
$\CQ(X)$ being an orthomodular lattice.  

\smnoind
(c) If $(X, \neq_\mq)$ is atomic,  then $\CQ(X)$ is an atomistic lattice, and one has $\CQ(X)^\mm = \big\{\{x\}: x\in X \big\}$.

\smnoind
(d) A bijection $\Phi:X\to Y$  is a strict quantum bijection if and only if it induces an ortholattice isomorphism from $\CQ(X)$ onto $\CQ(Y)$. 
\end{lem}
\begin{proof}
(a) It is easy to check that $^\mc$ is an orthocomplementation on the complete lattice $\CQ(X)$. 

\smnoind
(b) It follows from the definition and part (a) above that $\CQ(T)\subseteq \CQ(X)$ whenever $T\in \CQ(X)$. 
Hence, $\neq_\mq$ is hereditary if and only if for each $T\in \CQ(X)$, one has 
$$\{S\in \CQ(X): S\subseteq T \} \subseteq \CQ(T),$$ 
which is the same as $S = T\cap (T \cap S^\mc)^\mc$ for every $S\in \CQ(X)$ with $S\subseteq T$. 
This later condition is equivalent to 
$$S^\mc = T^\mc \vee (S^\mc\cap T), \ \text{for any $S,T\in \CQ(X)$ with $S\subseteq T$};$$ 
which is precisely the orthomodular law for $\CQ(X)$. 
This means that $\neq_\mq$ is hereditary if and only if $\CQ(X)$ is an orthomodular lattice. 
On the other hand, it is well-known that $\CQ(X)$ is an orthomodular lattice if and only if for every $S, T\in \CQ(X)$ with $S\subseteq T$ and $T\cap S^\mc = \emptyset$, one has $S = T$ (see, e.g., Theorem 2 of \cite[\S 1.3]{Kal83}).

\smnoind
(c) This part is easy to verify. 

\smnoind
(d) The forward implication is obvious (note that $\vee$ is defined through the q-complement). 
For the backward implication, we note that 
if $x\in X$, then $\Phi(\{x\}^\mc) = \Phi(\{x\}^{\mc\mc\mc}) = \Phi(\{x\}^{\mc\mc})^\mc \subseteq \{\Phi(x)\}^\mc$. 
This means that for any $z\in X$ with $z\neq_\mq x$, one has $\Phi(z) \neq_\mq \Phi(x)$. 
By considering $\Phi^{-1}$, we know that $\Phi$ is a strict quantum bijection.  
\end{proof}

\medskip

Note that in part (d) above, we do not assume any one of the two quantum sets to be atomic.



\medskip

\begin{eg}
Let $\KH$ be a Hilbert space. 
For any $\xi,\eta\in \KH^\star := \KH\setminus \{0\}$, we set $\xi \neq_\mo \eta$ if $\xi$ is orthogonal to $\eta$. 
Then $E\in \CQ(\KH^\star, \neq_\mo)$ if and only if $E \cup \{0\}$ is a closed subspace of $\KH$.
Thus, $\neq_\mo$ is not atomic.  
On the other hand, as $\CQ(\KH^\star, \neq_\mo)$ is isomorphic to the atomic orthomodular lattice of projections in $\CB(\KH)$, we know that $\neq_\mo$ is hereditary. 
This example also tells us that the converse of Lemma \ref{lem:lattice}(c) does not hold. 
\end{eg}

\medskip

Actually, a Hilbert space can be regarded as a Banach space $E$ with $E\setminus\{0\}$ being equipped with a compatible q-distinctness relation.

\medskip

\begin{eg}\label{eg:uncert}
For $u,v\in \BR$, we define 
$$u\neq_\mathrm{uc} v\quad \text{when} \quad |u-v| \geq 1.$$ 
If $u,v \in \BR$ with $u-v \geq 2$,  then $\{u,v\}^\mc = (-\infty, v-1]\cup [v+1, u-1] \cup [u+1, \infty)$ and 
$$\{u,v\}^{\mc\mc} = \{u,v\}.$$
On the other hand, if $S\subseteq \BR$ with $|x - y| < 2$ ($x,y\in S$), then  
\begin{equation}\label{eqt:n-points-dist>2delta}
S^\mc = (-\infty, x_0 -1] \cup [y_0+1, \infty) \quad \text{and} \quad S^{\mc\mc} = [x_0,y_0], 
\end{equation}
where $x_0:= \inf S$ and $y_0:=\sup S$. 
In particular, $\{u\}^{\mc\mc} = \{u\}$ and we know that $\neq_\mathrm{uc}$ is atomic.

We denote by $\CQ_1\subseteq \CP(\BR)$ the collection of subsets that are at most countable unions of  disjoint closed intervals (whose lengths could be zero, finite or infinite) such that any two elements in two distinct intervals is of distances bigger than or equal to $2$. 
For any $T\in \CQ_1$, it is not hard to check that $T^{\mc\mc} = T$.
This implies that $\CQ_1\subseteq \CQ(\BR, \neq_\mathrm{uc})$. 

Suppose that $S = \bigcup_{k\in \BZ} [x_k,y_k]$ with $y_k < x_{k+1}$. 
By grouping together those intervals with $x_{k+1} - y_k < 2$, we see that $S^\mc \in \CQ_1$ and hence $S^{\mc\mc}\in \CQ_1$. 

Let $D\subseteq \BR$. 
For any $n\in \BZ$, we set $D_n:= D\cap (2n, 2n+2)$. 
Observe that 
$$D^{\mc\mc} \supseteq {\bigcup}_{n\in \BZ} D_n^{\mc\mc} \cup (D\cap 2\BZ)  \supseteq D.$$ 
It follows from Relation \eqref{eqt:n-points-dist>2delta} that ${\bigcup}_{n\in \BZ} D_n^{\mc\mc} \cup (D\cap 2\BZ)$ is of the form $\bigcup_{k\in \BZ} [x_k,y_k]$ with $y_k < x_{k+1}$.
Hence, $\big({\bigcup}_{n\in \BZ} D_n^{\mc\mc} \cup (D\cap 2\BZ)\big)^{\mc\mc}$ belongs to $\CQ_1$, and thus, $D^{\mc\mc} =\big({\bigcup}_{n\in \BZ} D_n^{\mc\mc} \cup (D\cap 2\BZ)\big)^{\mc\mc}\in \CQ_1$. 
Consequently, $\CQ_1 = \CQ(\BR, \neq_\mathrm{uc})$. 

Consider $T=[0,1]$ and $S=[0,1/2]$. 
Then $S\subseteq T$ and $S,T\in \CQ_1$, but $S\notin \CQ(T)$ (as $S^\mc\cap T = \emptyset$). 
This means that $\neq_\mathrm{uc}$ is not hereditary, and hence, the complete atomistic ortholattice $\CQ(\BR, \neq_\mathrm{uc})$ is not orthomodular. 

Notice that if we replace $|u-v| \geq 1$ with $|u-v| \geq \delta$ for some fixed $\delta > 0$, then the resulting ortholattice will be isomorphic to the above. 
However, if we consider the q-distinctness relation induced by $|u-v| > 0$, then the resulting ortholattice is the classical Boolean algebra $\CP(\BR)$. 
\end{eg}

\medskip

Now, we will give the close connection between ortholattices and quantum sets, which shows that quantum sets can be viewed as the underlying ``set theory'' for quantum logic. 
The mathematics behinds these statements could be considered as known. 
Actually, the fact that the map $\Xi^0_\CL$ in part (a) is an ortholattice isomorphism for a complete ortholattice $\CL$ is a disguised form of \cite[Corollary 5.4]{Walker}, and part (b) is similar to the bijective correspondence between Hilbert lattices and Hilbert geometries (see \cite[Theorem 5.16]{SV07}). 
However, since these statements were not explicitly stated in the literature and the proof of part (b) is needed later on, we will give some words for their arguments here. 

\medskip

\begin{prop}\label{prop:atom-hered-q-distinct}
Let $\CL$ be an ortholattice, and $\CL^\star := \CL\setminus \{0\}$. 
For $p,q\in \CL^\star$, we set $p\neq_\CL q$ if $p\leq q'$. 

\smnoind
(a) Then $\neq_\CL$ is a q-distinctness relation on $\CL^\star$, and $\CL\mapsto (\CL^\star, \neq_\CL)$ is an injective correspondence from the collection of ortholattices to that of quantum sets. 
Moreover, the assignment 
$$p\mapsto \Xi^0_\CL(p):= \big\{q\in \CL^\star: q\leq p \big\}$$
is an injective ortholattice homomorphism from $\CL$ to $\CQ(\CL^\star, \neq_\CL)$. 
If, in addition, $\CL$ is complete, then the map $\Xi^0_\CL$ is a surjection (and in this case, $\CL$ is orthomodular if and only if $\neq_\CL$ is hereditary). 

\smnoind
(b) Suppose that $\CL$ is atomistic with $\CL^\mm\subseteq \CL^\star$ being the set of all atoms. 
The assignment
$$p\mapsto \Xi_\CL(p):= \{a\in \CL^\mm: a\leq p \}$$
is an injective ortholattice homomorphism from $\CL$ to $\CQ(\CL^\mm, \neq_\CL)$. 
Moreover, $\CL \mapsto (\CL^\mm, \neq_\CL)$ produces a bijective correspondence between the collections of complete atomistic ortholattices and atomic quantum sets (respectively, complete atomic orthomodular  lattices and atomic hereditary quantum sets). 
\end{prop}
\begin{proof}
(a) Obviously,  $\neq_\CL$ is a  q-distinctness relation on $\CL^\star$. 
Moreover, as $\Xi^0_{\CL}(p')= \Xi^0_{\CL}(p)^\mc$, one has $\Xi^0_{\CL}(p) = \Xi^0_{\CL}(p'')  = \Xi^0_{\CL}(p)^{\mc\mc} \in \CQ(\CL^\star)$. 
On the other hand, since for each $S \subseteq \CL^\star$, one has 
$$S^\mc = \{q\in \CL^\star: s\leq q', \text{ for any }s\in S \},$$ 
the embedding $\Xi^0_{\CL}: \CL\to \CQ(\CL^\star)$ is a disguised form of the MacNeille completion of $\CL$ (see e.g., \cite{MacN} or \cite[p.255]{Kal83}). 
Hence, $\Xi^0_{\CL}$ is an injective ortholattice homomorphism. 

In order to verify that the assignment $\CL\mapsto (\CL^\star, \neq_\CL)$ is an injective correspondence, we suppose that $\CM$ is another ortholattice and there is a strict quantum bijection $\Phi: (\CM^\star, \neq_{\CM}) \to (\CL^\star, \neq_\CL)$. 
Let $\bar \Phi: \CQ(\CM^\star)\to \CQ(\CL^\star)$ be the induced ortholattice isomorphism. 
Since $\Xi^0_{\CL}$ is an injective ortholattice homomorphism, if the assignment $p \mapsto (\Xi^0_{\CL})^{-1}\big(\bar\Phi\big(\Xi^0_{\CM}(p)\big) \big)$ is well-defined, then it will be an injective ortholattice homomorphism from $\CM$ to  $\CL$, and by symmetry, this homomorphism is bijective. 

To show $p \mapsto (\Xi^0_{\CL})^{-1}\big(\bar\Phi\big(\Xi^0_{\CM}(p)\big) \big)$ being well-defined, we observe that 
for every $u,p\in \CM$, one has $u\leq p$ if and only if $\Phi(u)\leq \Phi(p')'$, because $\Phi$ is a strict quantum bijection. 
In other words, $\bar \Phi(\Xi^0_{\CM}(p)) = \Xi^0_{\CL}(\Phi(p')')$ ($p\in \CM$). 
The bijectivity of $\Phi$ will then give the required relation  
$$\{\bar \Phi (\Xi^0_{\CM}(p)): p\in \CM^\star \} = \{\Xi^0_{\CL}(q): q\in \CL^\star \}.$$

Finally, when $\CL$ is complete, the map $\Xi^0_{\CL}$ is an ortholattice isomorphism because $\CQ(\CL^\star)$ is the  MacNeille completion of $\CL$.
The last statement follows from Lemma \ref{lem:lattice}(b). 

\smnoind
(b) In the following the q-complementation $^\mc$ is taken in $\CP(\CL^\mm)$. 
As in part (a), we have $\Xi_\CL(p)^\mc = \Xi_\CL(p')$, which gives $\Xi_\CL(p)\in \CQ(\CL^\mm)$. 
It is not hard to verify that $\Xi_\CL$ is an injective ortholattice homomorphism. 

Suppose, in addition, that the lattice $\CL$ is complete. 
Consider $S\in \CQ(\CL^\mm)$. 
Obviously, $v_0:= \bigvee S^\mc \leq a'$ for each $a\in  S$, which gives 
\begin{equation}\label{eqt:u-lequantum v-bot}
u_0:= \bigvee S\leq v_0'.
\end{equation}
On the other hand, it follows from $S\subseteq \Xi_\CL(u_0)$, $S^\mc \subseteq \Xi_\CL(v_0)$ and $S^{\mc\mc} = S$ that
$\Xi_\CL(v_0') \subseteq \Xi_\CL(u_0)$. 
Thus, by Relation \eqref{eqt:u-lequantum v-bot}, we have $S = \Xi_\CL(u_0)$. 
Hence, $\Xi_\CL$ is surjective.

The above, together with parts (a) and (c) of Lemma \ref{lem:lattice}, establishes that $\CL \mapsto (\CL^\star, \neq_\CL)$ is a bijective correspondence from the collection of complete atomistic ortholattices to that of atomic quantum sets. 
Moreover, the corresponding statement for orthomodular lattices follows from Lemma \ref{lem:lattice}(b). 
\end{proof}

\medskip

Note that the quantum set $\CL^\star$ in Proposition \ref{prop:atom-hered-q-distinct}(a) is never atomic, unless $\CL = \{0,1\}$ because $\{1\}^{\mc\mc} = \CL^\star$. 
In the case when $\CL\neq \{0,1\}$ and $\CL$ is atomic, the two distinct quantum sets $(\CL^\star, \neq_\CL)$ and $(\CL^\mm, \neq_\CL)$ give the same ortholattice $\CL$.

\medskip

Observe also that, without the completeness assumption, the quantum set $(\CL^\star, \neq_\CL)$ may not be hereditary even when $\CL$ is  orthomodular  (since the MacNeille completion of an orthomodular lattice need not be orthomodular; see e.g. \cite{Adams}).


\medskip

\medskip

\begin{eg}\label{eg:reg-open}
(a) Let $A$ be a $C^*$-algebra, and $A^\star := A\setminus \{0\}$. 
For $a,b\in A^\star$, we define 
$$a\neq_\mathrm{z} b\quad \text{if} \quad ab=0=ba \text{ and }a^*b = 0 =ba^*.$$
A subset of $A$ belongs to $\CQ(A^\star, \neq_\mathrm{z})$ if and only if it is an \emph{annihilator hereditary $C^*$-subalgebra} of $A$ (i.e., subalgebra of the form $\{a\in A: aB = \{0\} = Ba\}$ for a $C^*$-subalgebra $B\subseteq A$) minus the zero element. 
Hence, $\neq_\mathrm{z}$ is never atomic. 
On the other hand, since for an annihilator hereditary $C^*$-subalgebra $B$ of $A$, every annihilator hereditary $C^*$-subalgebra of $A$ that contained in $B$ is an 
annihilator hereditary $C^*$-subalgebra of $B$, we know that $\neq_\mathrm{z}$ is hereditary. 

When $A$ is commutative, $\CQ(A^\star, \neq_\mathrm{z})$ corresponds bijectively to the set of ``regular open subsets''  (i.e., interiors of closed subsets) of the Gelfand spectrum of $A$.

\smnoind
(b) Let $X$ be a compact Hausdorff space. 
It is well-known that the collection $\CR(X)$ of all regular open subsets of $X$ is a complete Boolean algebra with 
$$U\wedge V := U\cap V, \quad U\vee V := \mathrm{Int}(\overline{U}\cup \overline{V})\quad \text{and} \quad U':= X\setminus \overline{U}$$  
(see \cite[p.103-104]{Vlad}), and it coincides with the orthomodular lattice $\CQ(C(X)^\star, \neq_\mathrm{z})$ as in part (a).   
Since every ultrafilter on $\CR(X)$ converges to a unique point in $X$, one obtains a continuous open surjection from the Stone space of $\CR(X)$, denoted by $X^\varepsilon$, onto $X$. 
Note that this construction is different from the ``Stone space associated with $X$'' as introduced in \cite{Mihara}, where the collection of clopen subsets instead of regular open subsets were considered.  

It is not hard to check that $X^\varepsilon$ is the \emph{universal extremally disconnected space associated with  $X$}, in the sense that any continuous map from an extremally disconnected space $Y$ to $X$ can be lifted to a continuous map from $Y$ to $X^\varepsilon$.
Furthermore, it can be shown that $C(X^\varepsilon)$ is the regular monotone completion of the $C^*$-algebra $C(X)$, as introduced in \cite{Ham81}. 

Suppose that $Z$ is a compact Hausdorff space. 
If $\Phi$ is a strict quantum bijection from $\big(\CR(X)^\star, \neq_{\CR(X)}\big)$ onto $\big(\CR(Z)^\star, \neq_{\CR(Z)}\big)$, then one can use Proposition \ref{prop:atom-hered-q-distinct}(a) and the Stone representation theorem to obtain a homeomorphism  from $Z^\varepsilon$ onto $X^\varepsilon$ that induces $\Phi$. 
\end{eg}

\medskip

In order to define quantum topology (see Definition \ref{defn:quantum top}), we also need the notion of q-commutativity.

\medskip

\begin{defn}\label{defn:q-comm}
Let $\CL$ be an ortholattice and $p,q\in \CL$. 

\smnoind
(a) We say that $p$ \emph{q-commutes}  with $q$  if $p\wedge(p\wedge q)' \leq (q\wedge (p\wedge q)')'$.

\smnoind
(b) $p\in \CL$ is said to be \emph{q-central} if it q-commutes with all other elements in $\CL$. 
\end{defn}

\medskip

For a quantum set $(X, \neq_\mq)$, we may extend the notion of q-commutativity to general subsets: 
two subsets $C,D\in \CP(X)$ are said to be \emph{q-commutes} if 
$$C\cap (C\cap D)^\mc \subseteq (D\cap (C\cap D)^\mc)^\mc.$$
Notice that if either $C\subseteq D$ or $C\subseteq D^\mc$, then $C$ will q-commute with $D$. 
Hence, if $x\in X$, then the singleton subset $\{x\}$ q-commutes with $D$ if and only if either $x\in D$ or $x\in D^\mc$. 
This gives the following statements: 
\begin{itemize}
	\item a subset $D\subseteq X$ q-commutes with all singleton subsets of $X$ if and only if $D^\mc = X\setminus D$;
	
	\item for any $x,y\in X$, one has $\{x\}$ q-commutes with $\{y\}$ if and only if either $x=y$ or $x\neq_\mq y$. 
\end{itemize}
The second statement above, together with Proposition \ref{prop:atom-hered-q-distinct}(b),  tells us that 
a complete atomistic ortholattice is a Boolean algebra if and only if all its atoms q-commute with one another.

\medskip

Let us also recall from \cite[p.20]{Kal83} that $p$ is said to \emph{commute with} $q$ in an ortholattice if 
$$p = (p\wedge q)\vee (p\wedge q').$$ 
Note that this commutativity relation is, in general, asymmetric. 
It is natural to ask whether there is any relation between this commutativity and the q-commutativity above. 

\medskip

\begin{prop}\label{prop:cp-comm}
Let $\CL$ be an ortholattice. 
Then $\CL$ is orthomodular if and only if the q-commutativity relation coincides with the commutativity relation on $\CL$. 
\end{prop}
\begin{proof}
$\Rightarrow)$. 
Consider $p,q\in \CL$. 
Suppose that $p$ q-commutes with $q$. 
As $\CL$ is orthomodular, we have $q = (p\wedge q)\vee (q\wedge (p\wedge q)')$. 
From this, we know that
	$$q' = (p\wedge q)' \wedge (q\wedge (p\wedge q)')' \geq (p\wedge q)' \wedge (p\wedge (p\wedge q)')  = p\wedge (p\wedge q)'.$$
Thus, $(p\wedge q)\vee (p\wedge q') \geq (p\wedge q)\vee (p\wedge (p\wedge q)') = p$ (again by the orthomodular law). 

Conversely, suppose that $p$ commutes with $q$. 
Then, by Lemma 1 of \cite[\S 1.4]{Kal83}, the sub-ortholattice generated by $p$ and $q$ is distributive. 
Consequently, 
$$p\wedge (p\wedge q)' = p \wedge (p' \vee q') = (p\wedge p')\vee (p\wedge q') = p\wedge q'.$$
Similarly, $q\wedge (p\wedge q)' = q \wedge p'$ and hence, $p\wedge (p\wedge q)' \leq (q\wedge (p\wedge q)')'$. 

\noindent
$\Leftarrow)$. 
Note that the q-commutativity relation is symmetric. 
However, it was shown in Theorem 2 of \cite[\S 1.3]{Kal83} that if the commutative relation is symmetric, then $\CL$ is orthomodular. 
\end{proof}

\medskip

\begin{cor}\label{cor:q-cent-hered}
Let $(X, \neq_\mq)$ be a hereditary quantum set. 

\smnoind
(a) If two elements $S$ and $T$ in $\CQ(X)$ q-commute, then the sub-ortholattice generated by $S$ and $T$ is distributive and, in particular,  
$S$ and $T^\mc$ q-commute. 

\smnoind
(b) $S\in \CQ(X)$ is q-central if and only if $X\setminus (S\cup S^\mc)$ contains no non-empty element in $\CQ(X)$. 

\smnoind
(c) Suppose, in addition, that $\neq_\mq$ is atomic. 
Then a subset $D\subseteq X$ belongs to $\CQ(X)$ and is q-central if and only if $D^\mc = X \setminus D$. 
\end{cor}
\begin{proof}
(a) This follows from Proposition \ref{prop:cp-comm} and Lemma 1 of \cite[\S 1.4]{Kal83} (see also Lemma \ref{lem:lattice}(b)).

\smnoind
(b) Lemma \ref{lem:lattice}(b) implies that $\CQ(X)$ is an orthomodular lattice. 
Suppose that $S$ is q-central. 
Consider $T\in \CQ(X)$ with $T\subseteq (X\setminus S)\cap (X\setminus S^\mc)$. 
As $T$ commutes with $S$ (by Proposition \ref{prop:cp-comm}), we see that 
$T  = (T\cap S)\vee (T\cap S^\mc)$, but $T\cap S = T\cap S^\mc = \emptyset$. 
Conversely, suppose that the only element in $\CQ(X)$ contained in $X\setminus (S\cup S^\mc)$ is the empty set. 
Consider $T\in \CQ(X)$. 
Obviously, $T \supseteq (T\cap S)\vee (T\cap S^\mc)$.
We set 
$$R:= T \cap ((T\cap S)\vee (T\cap S^\mc))^\mc.$$ 
Then $R\cap S = T\cap S \cap (T\cap S)^\mc \cap (T\cap S^\mc)^\mc = \emptyset$, and similarly, we have $R\cap S^\mc = \emptyset$. 
Thus, $R \subseteq (X\setminus S)\cap (X\setminus S^\mc)$ and the hypothesis implies that $R =\emptyset$. 
As $\CQ(X)$ is orthomodular, we see that $T = (T\cap S)\vee (T\cap S^\mc)$; i.e., $T$ commutes with $S$. 
The conclusion now follows from Proposition \ref{prop:cp-comm}. 

\smnoind
(c) If $D\in \CQ(X)$ and is q-central, then part (b) and the atomic assumption give $D^\mc = X \setminus D$. 
Conversely, suppose that $D^\mc = X \setminus D$. 
Then it is easy to see that $D\in \CQ(X)$, and we conclude from part (b) that $D$ is q-central. 
\end{proof}

\medskip

\begin{eg}\label{eg:tran-prob}
	Let $\BR\BP^1$ be the one dimensional projective space.
	The absolute value of the usual inner product on $\BR^2$ induces a function $\tau: \BR\BP^1\times \BR\BP^1 \to [0,1]$ known as the transition probability. 
	We define
	$$x \neq_\mathrm{ht}y \quad \text{when} \quad \tau(x,y) \leq 1/2.$$  
	Obviously, $\neq_\mathrm{ht}$ is a q-distinctness relation. 
	We identify $\BR\BP^1$  with the interval $[0,3)$ as sets  (i.e., not homeomorphic) through the assignment that sends $\theta$ to the image of $\mathrm{e}^{\theta\pi\mathrm{i}/3}$ in $\BR\BP^1$. 
	Define $\kappa: [0,6) \to [0,3)$ to be the map which coincides with the identity map on $[0,3)$ and $\kappa (\theta) := \theta - 3$ when $\theta \geq 3$. 
	Then 
	$$\{\theta\}^\mc = \kappa\big([\theta+ 1, \theta+ 2]\big) \quad \text{and} \quad \{\theta\}^{\mc\mc} = \{\theta\} \qquad (\theta\in [0,3)).$$ 
	This means that $\neq_\mathrm{ht}$ is atomic. 
	If $\theta_1, \theta_2\in [0,6)$ with $\theta_2 - \theta_1 \in [0,3)$, then we say that  $\kappa\big([\theta_1, \theta_2]\big)$ is an \emph{arc of length $\theta_2 - \theta_1$} (note that an arc of length zero is a singleton subset). 
	
	Pick  an arbitrary non-empty subset $D\subseteq \BR\BP^1$. 
	If $D$ is contained in an arc of length $1$, then $D^{\mc\mc}$ is the smallest arc containing $D$; otherwise, we have $D^{\mc\mc} = \BR\BP^1$. 
	Therefore, 
	$$\CQ\big(\BR\BP^1, \neq_\mathrm{ht}\!\!\big) = \big\{\emptyset, \BR\BP^1 \big\}\cup \big\{ S: S\subseteq \BR\BP^1 \text{ is an arc of length dominated by }1 \big\}. $$
	Consider $T = \kappa\big([0, 1]\big)$ and $S := \kappa\big([0, 1/2]\big)$. 
	Then $S^\mc \cap T = \emptyset$ and hence $(S^\mc \cap T)^\mc  \cap T = T\neq S$. 
	In other words, $\neq_\mathrm{ht}$ is not hereditary.
		
Suppose that $S, T\in  \CQ(\BR\BP^1, \neq_\mathrm{ht})\setminus \{\emptyset, \BR\BP^1 \}$. 
If $S$ is an arc of length strictly less than one, then $S$ commutes with $T$ if and only if either $S\subseteq T$ or $S\subseteq T^\mc$.
If $S$ is an arc of length one, then $S$ commutes with $T$ if and only if either $S = T$, $S = T^\mc$ or both $S\cap T$ and $S\cap T^\mc$ are singleton subsets. 
However,  $S$ q-commutes with $T$ if and only if either $S\cap T\neq \emptyset$ or $S\subseteq T^\mc$.  
In particular, $\emptyset$ and $\BR\BP^1$ are the only q-central elements in $\CQ(\BR\BP^1, \neq_\mathrm{ht})$. 
\end{eg}

\medskip

It is more interesting to look at the case when the q-distinctness relation is defined by elements having zero transition probability. 
This will be considered in Section \ref{sec:non-comm-Gelf} below, but in a much more general situation. 
As said in the Introduction, Section \ref{sec:non-comm-Gelf} contains a non-commutative version of the Gelfand theorem, and this theorem requires a version of Dye's theorem for $C^*$-algebras. 
Thus, we will consider this in the next section.

\medskip

\section{A Dye's theorem for $C^*$-algebras}

\medskip

In this section, we will present a Dye theorem for $C^*$-algebras concerning q-closed projections. 
Let us first recall the original Dye's theorem. 
We denote by $\CP_M$ the set of projections of a von Neumann algebra $M$. 
A linear map (respectively, linear bijection) $\Theta$ from a $C^*$-algebra $A$ to another $C^*$-algebra is a \emph{Jordan $^*$-homomorphism} (respectively, \emph{Jordan $^*$-isomorphism}) if $\Theta$ preserves the involution and the Jordan product; i.e., $\Theta(a^*) = \Theta(a)^*$ and $\Theta(ab + ba)= \Theta(a)\Theta(b) + \Theta(b)\Theta(a)$ ($a,b\in A$). 

\medskip

The original Dye's theorem (\cite{Dye}) can be rewritten in the following form. 

\medskip

%

\medskip

\begin{thm2}\label{thm-Dye}
	(Dye) Let $M$ and $N$ be two von Neumann algebras with $M$ not having a type $\mathrm{I}_2$ summand. 
	Any strict quantum bijection from $\CP_M\setminus \{0\}$ onto $\CP_N\setminus \{0\}$ (under the q-distinctness relations induced by orthogonality) extends uniquely to a Jordan $^*$-isomorphism from $M$ onto $N$. 
\end{thm2}

\medskip

There has been extensions of this theorem to the case of $JW$-algebras (see \cite{BW93}) and $AW^*$-algebras (see \cite{Ham}). 
However, we need an ``extension'' to $C^*$-algebra. 
\medskip

Before presenting our version of Dye's theorem for $C^*$-algebra that is required, let us first give some notation. 
We denote by $\widehat{A}$ the set of (unitary equivalence classes) of irreducible $^*$-representations of $A$. 
An element $p\in \CP_{A^{**}}$ is a \emph{closed projection} of the  $C^*$-algebra $A$ (where $A^{**}$ is the enveloping von Neumann algebra of $A$) if there is an increasing net $\{a_i \}_{i\in \KI}$ of positive contractive elements in $A$ such that $1-a_i$ weak-$^*$-converges to $p$ (see e.g. \cite[\S 3.11.10]{Ped79}). 
An element $q\in \CP_{A^{**}}$ is an \emph{open projection} of $A$ if $1-q$ is a closed projection of $A$; i.e., there exists an increasing net $\{a_i\}_{i\in \KI}$ in $A_+$ such that $a_i$ weak-$^*$-converges to $q$. 
We denote by $\CC_0(A)$ the set of all closed projections of $A$. 

\medskip


\medskip

As in the previous section, $\CP^\mm_{A^{**}}$ is the set of atoms of the complete orthomodular lattice $\CP_{A^{**}}$.
We denote by $A^{**}_\ba$ the weak-$^*$-closed linear span of $\CP^\mm_{A^{**}}$. 
Furthermore, we define $\bz_A:= \bigvee \CP_{A^{**}}^\mm\in \CP_{A^{**}}$ and consider the normal $^*$-homomorphism $\Lambda_A:A^{**} \to A^{**}_\ba$ given by 
\begin{equation}\label{eqt:def-Lambda}
\Lambda_A(x) := \bz_Ax \qquad (x\in A^{**}).
\end{equation}
Since the restriction of $\Lambda_A$ on the multiplier algebra $M(A)$ is injective (when $M(A)$ is considered as a subspace of $A^{**}$), by abuse of notation, we will consider $M(A)$ as a $C^*$-subalgebra of $A^{**}_\ba$. 

\medskip

The image of a closed projection (respectively, an open projection) under $\Lambda_A$ is  called a \emph{q-closed projection} (respectively, a  \emph{q-open projection}). 
The set of all q-closed projections of $A$ will be denoted by $\CC(A)$. 

\medskip

Before we start proving the main theorem (Theorem \ref{thm:main}) of this section, we first show that the corresponding statement of this theorem with ``q-closed projections'' being replaced by ``q-open projections''  is false (similarly, Corollary \ref{cor:closed-proj} fails when we consider open projections instead of closed projections).

\medskip

\begin{eg}\label{eg:counter}
	Note that $C([0,1])^{**}_\ba \cong \ell^\infty([0,1])$.
	A projection $p\in \ell^\infty([0,1])$ is a q-open projection of $C([0,1])$ if and only if $p$ is the indicator function of an open subset of $[0,1]$. 
	Consider $\CO$ to be the collection of all open subsets of $[0,1]$. 
	Define $\varphi: \CO \to \CO$ such that 
	$\varphi$ exchanges $[0,1)$ and $(0,1]$, but $\varphi$ fixes all the other open subsets.
	Then $\varphi$ induces a strict quantum bijection $\hat\varphi$ from the set of all non-zero q-open projections of $C([0,1])$ onto itself. 
	However,  $\hat \varphi$ cannot be the restriction of the bidual of a $^*$-automorphism of $C([0,1])$, because a $^*$-automorphism of $C([0,1])$ is given by a homeomorphism on $[0,1]$, yet $\varphi$ cannot be induced by a homeomorphism. 
\end{eg}

\medskip

More generally, let $X$ be a compact Hausdorff space.
If $U\subseteq X$ is an open dense subset, then there is no non-trivial open subset having empty intersection with $U$. 
Thus, as in the example above, if a map permutes open dense subsets of $X$ but keeps all the other open subsets fixed, then this map induces a strict quantum bijection between non-zero q-open projections but this strict quantum bijection does not come from a homeomorphism of $X$ to itself. 
One may wonder what will happen if the collection of regular open subsets is considered instead (in this case, those problematic open dense subsets are not included).  
In Example \ref{eg:reg-open}, we have seen that in this case, we do not get a homeomorphism from $X$ to itself, but a homeomorphism from the spectrum of the regular monotone completion of $C(X)$ to itself. 
Therefore, there seems to have no way to get a Dye's theorem for q-open projections of $C^*$-algebras.





\medskip

Let us begin our proof for the main theorem in the section by giving some basic results. 
Our first proposition in the section is an ``atomic version'' of \cite[Theorem 2.2]{APT73}. 
Although this fact could be regarded as known, we nevertheless give a clear presentation of it for completeness. 

\medskip

\begin{prop}\label{prop:atom-Thm2.2-APT}
	Let $A$ be a $C^*$-algebra and $x\in A^{**}_\ba$ be a self-adjoint element.
	Then $x\in M(A)$ if and only if all its spectral projections (in $A^{**}_\ba$) corresponding to closed subsets of $\BR$ being q-closed. 
\end{prop}
\begin{proof}
	If $x\in M(A)_\sa$, then it follows from \cite[Theorem 2.2]{APT73} that all its spectral projections (in $A^{**}_\ba$) corresponding to closed subsets of $\BR$ are q-closed. 

	Conversely, suppose that such a property holds for $x\in (A^{**}_\ba)_\sa$. 
	Let $\CU_0\subseteq A^{**}_\sa$ be the monotone sequential closed (real) Jordan algebra as in \cite[Proposition 3.6]{Ped72}. 
	It was shown in \cite[Theorem 4]{Brown14} that $\CU_0$ is the self-adjoint part of a $C^*$-algebra. 
	Moreover, as noted in the statement preceding \cite[Lemma 3.5]{Ped72}, $\CU_0$ contains the weak-$^*$-limits of all increasing nets in $A_\sa$. 
	Hence, $\CU_0$ contains all the open projections of $A$. 
	
	For any Borel subset $S\subseteq \BR$, we denote by $\chi_S$ the indicator function of $S$. 
	Consider $\alpha, \beta, \gamma \in \BR$ with $\alpha < \beta < \gamma$. 
	By the hypothesis, $\chi_{(\alpha,\gamma)}(x)$ is q-open.
	We fix an open projection $p_{(\alpha, \gamma)}\in \CP_{A^{**}}$ satisfying $\Lambda_A(p_{(\alpha, \gamma)}) = \chi_{(\alpha,\gamma)}(x)$.
	As $\chi_{(\alpha, \beta]}(x) =  \chi_{(\alpha, \gamma)}(x) - \chi_{(\beta, \gamma)}(x)$, the element $p_{(\alpha, \beta]} := p_{(\alpha,\gamma)} - p_{(\beta,\gamma)}$ will satisfy
	$$\Lambda_A(p_{(\alpha, \beta]}) = \chi_{(\alpha,\beta]}(x).$$
	
	Let us consider a sequence $\{x_n\}_{n\in \BN}$ in $A^{**}_\ba$ whose members are real linear spans of elements of the form $\chi_{(\alpha, \beta]}(x)$ such that $\|x-x_n\| \to 0$. 
	The above then produces a sequence $\{y_n\}_{n\in \BN}$ in $\CU_0$ with $\Lambda_A(y_n) = x_n$. 
	Since the restriction of $\Lambda_A$ on $\CU_0$ is isometric (see \cite[Theorem 3.8]{Ped72}), we know that the sequence $\{y_n\}_{n\in \BN}$ is Cauchy in $\CU_0$, and hence it converges to an element $y\in \CU_0$ that satisfies $\Lambda_A(y) = x$. 
	
	Pick an arbitrary open subset $O\subseteq (-\|x\| - 1, \|x\| + 1)$. 
	There exists an increasing sequence $\{f_n\}_{n\in \BN}$ in $C_0(\BR)_+$ that converges pointwisely to $\chi_O$. 
	As $f_n(y)\in \CU_0$ and $\CU_0$ is monotone sequential closed, we know that $\chi_O(y)$ belongs to $\CU_0$. 
	Since the $^*$-homomorphism $\Lambda_A$ is weak-$^*$-continuous, we have $\Lambda_A\big(\chi_O(y)\big) = \chi_O(x)$. 
	On the other hand, the hypothesis tells us that there is an open projection $q$ of $A$ with $\Lambda_A(q) = \chi_O(x)$. 
	Now, because both $\chi_O(y)$ and $q$ belongs to $\CU_0$ and $\Lambda_A$ restricts to an injection on $\CU_0$, we know that $q = \chi_O(y)$. 
	In other words, all the spectral projections of $y$ in $A^{**}$ with respects to open subsets of $\BR$ are open projections of $A$.
	Therefore, \cite[Theorem 2.2]{APT73} tells us that $y$ belongs to the canonical image of $M(A)$ in $A^{**}$. 
	Hence, $x$ is in the canonical image of $M(A)$ in $A^{**}_\ba$.
\end{proof}


\medskip

\begin{lem}\label{lem:C-st-alg}
	Let $A$ and $B$ be $C^*$-algebras. 
	Suppose that $\Phi:A^{**}_\ba \to B^{**}_\ba$ is a weak-$^*$-continuous unital Jordan $^*$-homomorphism.
	
	\smnoind
	(a) If $\Phi(\CC(A))\subseteq \CC(B)$, then $\Phi\big(M(A))\subseteq M(B)$. 
	
	\smnoind
	(b) If $\Phi$ is bijective and $\Phi(M(A)) =  M(B)$, then $\Phi\big(A)= B$.
\end{lem}
\begin{proof}
	(a) Consider $a\in M(A)_\sa$. 
	The unital von Neumann subalgebra $W^*(a,1)$ of $A^{**}_\ba$ generated by $a$ is commutative. 
	Hence, $\Phi$ restricts to a weak-$^*$-continuous unital $^*$-homomorphism from $W^*(a,1)$ to $W^*\big(\Phi(a),1\big)$. 
	Now, it follows from Proposition \ref{prop:atom-Thm2.2-APT} that $\Phi(a)\in M(B)$.
	
	\smnoind
	(b) Denote by  $\iota_A:M(A)\to A^{**}_\ba$ the canonical embedding, and consider $\bar \iota_A: M(A)^{**}_\ba \to A^{**}_\ba$ to be the restriction of the weak-$^*$-continuous $^*$-homomorphism $\ti\iota_A: M(A)^{**} \to A^{**}_\ba$ extending $\iota_A$. 
	It is not hard to see that the support of $\bar \iota_A$ coincides with the unique central q-open projection  $q_A$ of $M(A)$ that satisfies $A = q_AM(A)^{**}_\ba q_A\cap M(A)$. 
	Define $\Psi(x) := \iota_B^{-1}\circ \Phi\circ \iota_A(x)\in M(B)$ ($x\in M(A)$).
	Then $\ti \iota_B\circ \Psi^{**}  = \Phi\circ \ti \iota_A$ and 
	$$\bar \iota_B\circ \Psi^{**}|_{M(A)^{**}_\ba} = \Phi\circ \bar \iota_A.$$
	Thus, $\bar \iota_B(\Psi^{**}(q_A)) = 1$ which implies $\Psi^{**}(q_A)\geq q_B$.
	Hence, $\Psi(A)\supseteq B$; i.e., $B\subseteq \Phi(A)$.
	Similarly, we have $A\subseteq \Phi^{-1}(B)$. 
\end{proof}

\medskip

Note the difference between part (b) above and \cite[Proposition 2.3]{NW} that the atomic parts of the biduals are considered here (and so a different proof is required).

\medskip

\begin{thm}\label{thm:main}
	Let $A$ and $B$ be two $C^*$-algebras such that $A$ does not have a 2 dimensional irreducible $^*$-representation. 
	If $\Phi: \CC(A)\setminus \{0\} \to \CC(B)\setminus \{0\}$ is a strict quantum bijection (where the q-distinctness relations are the orthogonality relations), then there exists a Jordan $^*$-isomorphism $\Theta:A\to B$ with $\Phi = \Theta^{**}|_{\CC(A)}$. 
\end{thm}
\begin{proof}
As in Proposition \ref{prop:atom-hered-q-distinct}, we consider the q-distinctness relation $\neq_{\CP_{A^{**}_\ba}}$ on $\CP_{A^{**}_\ba}^\star$. 
Obviously, this q-distinctness relation extends the one on $\CC(A)\setminus \{0\}$ as considered in the statement. 
In the following, $^\mc$ is the q-complement in $\CQ\big(\CP_{A^{**}_\ba}^\star\big)$. 
Let us set 
\begin{equation}\label{eqt:def-CP[p]}
\CP^\mm_{A^{**}}[p]:=\{e\in \CP^\mm_{A^{**}}: e \leq p \} \qquad (p\in \CP_{A^{**}_\ba}).
\end{equation}
We recall that minimal projections in $A^{**}_\ba$ are q-closed (see e.g. \cite[Corollary 2]{Akemann68}). 
Thus, $\CP^\mm_{A^{**}}$ coincides with the set of minimal elements in the ordered subset $\CC(A)\setminus \{0\}$ of $\CP_{A^{**}_\ba}$. 
	
	Suppose that $r,s\in \CC(A)$ with $r\leq s$. 
	For any $f\in \{\Phi(s)\}^\mc\cap \CP_{B^{**}}^\mm$, one has $\Phi^{-1}(f)s = 0$ and hence  
	$$f\in \{\Phi(r)\}^\mc\cap \CP_{B^{**}}^\mm.$$ 
	As $\CP_{B^{**}_\ba}$ is an atomistic lattice,  the above implies that $\{\Phi(s)\}^\mc\subseteq \{\Phi(r)\}^\mc$, and hence, $\Phi(r)\leq \Phi(s)$. 
	This means that $\Phi$ is order preserving.
	Similarly, $\Phi^{-1}$ is order preserving.
	In particular, $\Phi(\CP^\mm_{A^{**}}) = \CP^\mm_{B^{**}}$.
	Thus, by the proof of Proposition \ref{prop:atom-hered-q-distinct}(b) (observe that $\CP_{A^{**}}^\mm[p] = \Xi_{\CP_{A^{**}}}(p)$), there is an ortholattice isomorphism $\Upsilon: \CP_{A^{**}_\ba}\to \CP_{B^{**}_\ba}$ satisfying 
	$$\Upsilon(p):= \bigvee \Phi\big(\CP^\mm_{A^{**}} [p]\big) \qquad (p\in \CP_{A^{**}_\ba}).$$
	
	Since $\Phi$ is an order isomorphism, if $r\in \CC(A)$, then 
	$\Phi\big(\CP^\mm_{A^{**}} [r]\big) = \CP^\mm_{B^{**}} [\Phi(r)],$
	and hence $\Upsilon(r) = \bigvee \CP^\mm_{B^{**}} [\Phi(r)] = \Phi(r)$. 
	This means that  $\Upsilon$ extends $ \Phi$. 
	Moreover, as $\Upsilon$ is an ortholattice isomorphism, we have 
	\begin{equation}\label{eqt:bij}
	\CP^\mm_{B^{**}}[\Upsilon(p)] =  \Upsilon\big(\CP^\mm_{A^{**}}[p]\big) =  \Phi\big(\CP^\mm_{A^{**}}[p]\big)  \qquad (p\in \CP_{A^{**}_\ba}).
	\end{equation}

	In the same way, the strict quantum bijection $\Phi^{-1}$ also extends to an ortholattice isomorphism $\Upsilon^\#: \CP_{B^{**}_\ba} \to \CP_{A^{**}_\ba}$. 
	By Relation \eqref{eqt:bij}, the map $\Upsilon^\#$ is the inverse of $\Upsilon$, which implies that $\Upsilon|_{\CP_{A^{**}_\ba}^\star}$ is a strict quantum bijection. 
	Therefore, it follows from Theorem \ref{thm-Dye} that $\Upsilon$ extends to a Jordan $^*$-isomorphism $\bar \Phi$ from $A^{**}_\ba$ onto $B^{**}_\ba$. 
	The assumption on $\Phi$ and Lemma \ref{lem:C-st-alg} now tells us that $\bar \Phi(A) = B$.
	
	We denote $\Theta:= \bar \Phi|_A$. 
	By the weak-$^*$-density of $A$ in $A^{**}_\ba$ and the automatic weak-$^*$-continuity of $\bar \Phi$, we know that $\bar \Phi\circ \Lambda_A = \Lambda_B\circ \Theta^{**}$.
	Finally, as $\bar \Phi$ extends $ \Phi$, we know  that $\Theta^{**}$ extends $\Phi$.  
\end{proof}

\medskip

\begin{rem}\label{rem:main-general-case}
	(a) As in the case of von Neumann algebra, when $A = B = \BM_2$, a strict quantum bijection from $\CC(A)\setminus \{0\} = \CP_A\setminus \{0\}$ onto $\CC(B)\setminus \{0\}$ may not comes from a Jordan $^*$-isomorphism from $A$ to $B$ (see Example \ref{eg:M2} in the next section for a clearer illustration on this). 
	
	\smnoind
	(b) The assumption on $A$ concerning irreducible $^*$-representations can be rephrased as $A$ not having $\BM_2$ as a quotient $C^*$-algebra; equivalently, for any closed ideal $I\subseteq A$ with $A/I$ not having a representation on $\BC$, there exist $x_1, x_2, x_3, x_4\in A/I$ such that 
	$${\sum}_{\sigma\in S_4} \mathrm{sgn}(\sigma) x_{\sigma(1)}x_{\sigma(2)}x_{\sigma(3)}x_{\sigma(4)}\neq 0,$$ where $S_4$ is the permutation group on 4 elements and $\mathrm{sgn}(\sigma)$ is the sign of $\sigma$ (see \cite[Proposition 2.3]{FR}). 
\end{rem}



\medskip

We also have the corresponding statement of Theorem \ref{thm:main} concerning the set $\CC_0(A)$ of all closed projections (instead of q-closed projections). 

\medskip

\begin{cor}\label{cor:closed-proj}
	Let $A$ and $B$ be two $C^*$-algebras such that $\BM_2$ is not a quotient $C^*$-algebra of $A$. 
	If $\Phi: \CC_0(A)\setminus \{0\} \to \CC_0(B)\setminus \{0\}$ is a strict quantum bijection (again, when the q-distinctness relations induced by orthogonality is considered), there is a Jordan $^*$-isomorphism $\Theta:A\to B$ such that $\Phi = \Theta^{**}|_{\CC_0(A)}$. 
\end{cor}
\begin{proof}
By \cite[Theorem II.17]{Ake69}, we know that $\Lambda_A$ induces a bijection from $\CC_0(A)$ onto $\CC(A)$ (notice that although the unital assumption is needed in \cite{Ake69}, the result \cite[Theorem II.17]{Ake69} still holds without the unital assumption). 
Consider $p,q\in \CC_0(A)$ with $\Lambda_A(p) \Lambda_A(q) = 0$. 
Then $\Lambda_A(p)\leq \Lambda_A(\mathbf{1}-q)$ and \cite[Theorem II.17]{Ake69} tells us that $p \leq \mathbf{1}-q$, because $\mathbf{1} - q$ is an open projection of $A$. 
This means that $\Lambda_A: \CC_0(A)\setminus \{0\} \to \CC(A)\setminus \{0\}$ is a strict quantum bijection. 
Now, the argument of Theorem \ref{thm:main} implies the required conclusion. 
\end{proof}

\medskip

\section{Quantum topologial spaces and Gelfand spectra of $C^*$-algebras}\label{sec:non-comm-Gelf}

\medskip

\begin{defn}\label{defn:quantum top}
Let $\CL$ be an ortholattice. 
A \emph{quantum topology} on $\CL$ is a subcollection $\CC\subseteq \CL$ satisfying:
\begin{enumerate}[\ \ S1).]
	\item $0, 1\in \CC$;
	\item if $\{{p_\lambda}\}_{\lambda\in \Lambda}$ is a family in $\CC$, then $\bigwedge_{\lambda\in \Lambda}p_\lambda$ exists and belongs to $\CC$;
	\item if $p$ and $q$ are q-commuting elements in $\CC$ (see Definition \ref{defn:q-comm}(a)), then $p\vee q\in \CC$.
\end{enumerate} 
In this case, elements in $\CC$ are said to be \emph{quantum closed}, while elements of the form $p'$ for some $p\in \CC$ are said to be \emph{quantum open}. 
\end{defn}	




\medskip 
	
\begin{defn}
(a)	Let $(X, \neq_\mq)$ be a quantum set.  
If $\CC\subseteq \CQ(X)$ is a quantum topology, then $(X, \neq_\mq, \CC)$ is called a \emph{quantum topological space}.

\smnoind
(b) A \emph{strict quantum homeomorphism} from a quantum topological space $(X, \neq_\mq, \CC)$ to another  quantum topological space $(Y, \neq_\mq, \mathcal{D})$ is a strict quantum bijection $\Psi:X\to Y$ satisfying $\mathcal{D} = \{\Psi(C): C\in \CC \}$. 
\end{defn}



\medskip

Observe that if $\neq_\mq$ is classical, then quantum topologies on $(X, \neq_\mq)$ coincides with usual topologies on $X$, and strict quantum homeomorphisms are precisely ordinary homeomorphisms. 
We will consider general quantum continuous maps between quantum topological spaces in a later work (\cite{Ng-Cat-QS}).

\medskip

In the following, we consider a particular kind of quantum topological spaces. 
Let $B$ be a $C^*$-algebra and $\KP^B$ be the set of all pure states on $B$. 
For $\phi\in \KP^B$, we denote by $\bs_\phi\in \CP_{B^{**}}$ the support projection of $\phi$ (i.e. $\bs_\phi$ is the smallest projection in $B^{**}$ with $\phi(\bs_\phi) = 1$). 
We equip $\KP^B$ with the q-distinctness relation: 
$$\phi\neq_\mo \psi \quad  \text{if and only if} \quad \bs_\phi\bs_\psi=0;$$ 
equivalently, $\phi\neq_\mo \psi$ means that the transition probability between $\phi$ and $\psi$ is zero.

\medskip

For a closed left ideal $L\subseteq B$, we set 
$$\hull(L):= \big\{\phi\in \KP^B: L \subseteq L_\phi \big\},$$
where $L_\phi:=\{b\in B: \phi(b^*b) =0 \}$. 
It is well-known that $L = \bigcap \{L_\phi: \phi \in \hull(L)\}$. 
Let us denote 
$$\CC^B:= \{\hull(L): L \text{ is a closed left ideal of }B \}.$$ 
We call $(\KP^B, \neq_\mo, \CC^B)$ the \emph{Gelfand spectrum} of $B$. 

\medskip

One can rephrase the Gelfand spectrum in terms of modular maximal left ideals of $B$ (thanks to \cite[Theorem 5.3.5]{Mur}). 
However, be aware that in this case, the correct q-distinctness relation \emph{is not} the one given by $L_1 L_2^* = \{0\}$. 

\medskip

In the following, we will show that $\CC^B$ is a quantum topology.
Indeed, it is clear that $\CC^B$ satisfies Conditions (S1) and (S2). 
In the following, we will verify that all elements in $\CC^B$ are q-subsets and that $\CC^B$ satisfies Condition (S3).  
To do these, we need the following facts. 

\medskip

\begin{lem}\label{lem:PS-min-proj}
Let $B$ be a $C^*$-algebra. 

\smnoind
(a) There is an order reversing bijection from $\CC(B)$ to the set of all closed left ideals of $B$, that associates $p\in \CC(B)$ with $L_p:= B^{**}_\ba(1-p) \cap B$.
In this case, $p\in \CC(B)$ is central if and only if $L_p$ is an ideal.

\smnoind
(b) The assignment $\phi\mapsto \bs_\phi$ is a strict quantum bijection from $\KP^B$ onto $\CP_{B^{**}}^\mm$, when $\CP_{B^{**}}^\mm$ is equipped with the q-distinctness relation induced from $\CP_{B^{**}_\ba}^\star$ (see Proposition \ref{prop:atom-hered-q-distinct}). 

\smnoind
(c) $p\mapsto \KP^B[p]:= \{\phi\in \KP^B: \bs_\phi\leq p \}$ is an ortholattice isomorphism from $\CP_{B^{**}_\ba}$ onto $\CQ(\KP^B)$.  

\smnoind
(d) $\CC^B = \big\{\KP^B[q]: q\in \CC(B) \big\}$. 

\smnoind
(e) For any $p,q\in \CP_{B^{**}_\ba}$, the two q-subsets $\KP^B[p]$ and $\KP^B[q]$ q-commute if and only if $pq = qp$. 
\end{lem} 
\begin{proof}
(a) This part is well-known (see, e.g., \cite[Theorem II.17]{Ake69} and \cite[\S 3.11.10]{Ped79}). 

\smnoind
(b) This part follows from the definitions of the two q-distinctness relations. 

\smnoind
(c) Note that $p\mapsto \CP_{B^{**}}^\mm[p]$ (see \eqref{eqt:def-CP[p]}) is an ortholattice isomorphism from $\CP_{B^{**}_\ba}$ onto $\CQ(\CP_{B^{**}}^\mm)$ because of the proof of Proposition \ref{prop:atom-hered-q-distinct}(b). 
The conclusion then follows from part (b) above.  

\smnoind
(d) This part follows from parts (a) and (b) as well as the fact that $L_\phi = L_{\bs_\phi}$, for any $\phi\in \KP^B$. 

\smnoind
(e) This follows from part (b) as well as the well-known fact that for any $p,q\in \CP_{B^{**}_\ba}$, one has $pq = qp$ if and only if $p- p\wedge q$ is orthogonal to $q - p\wedge q$.  
\end{proof}

\medskip

We also need the following ``atomic version'' of \cite[Theorem II.7]{Ake69}. 
This result actually follows from the argument of \cite[Theorem II.7]{Ake69} (recall that $\sigma(\bz_A\cdot A^*, A)$-closed left $A$-invariant subspaces of $\bz_A\cdot A^*$  are in bijective correspondence with closed left ideals of $A$; where $\bz_A$ is the central projection as in \eqref{eqt:def-Lambda}). 

\medskip

\begin{lem}\label{lem:sum-orth-closed}
	Let $A$ be a $C^*$-algebra. 
	If $p,q\in \CC(A)$ satisfying $\|p(q-p\wedge q)\|< 1$, then $p\vee q \in \CC(A)$. 
\end{lem}

\medskip

\begin{prop}\label{prop:quantum top}
Let $B$ be a $C^*$-algebra.

\smnoind
(a) The q-distinctness relation $\neq_\mo$ on $\KP^B$ is both atomic and hereditary. 

\smnoind
(b) $(\KP^B, \neq_\mo, \CC^B )$ is a quantum topological space. 

\smnoind
(c) Suppose that $\KG: \KP^B\to \widehat{B}$ is the surjection that sends $\omega\in \KP^B$ to the equivalence class $[\pi_\omega]_\sim$ of its GNS construction $\pi_\omega$. 
Then 
$$\big\{\KG^{-1}([\pi]_\sim): [\pi]_\sim\in\widehat{B} \big\}$$ 
is the collection of all minimal q-central elements in $\CQ(\KP^B)$. 
Moreover, $Z\mapsto \KG^{-1}(Z)$ is a bijective correspondence from the collection of closed (respectively, open) subsets of $\widehat{B}$ to the collection of q-central elements in $\CC^B$ (respectively, q-central quantum open subsets of $\KP^B$).
\end{prop}
\begin{proof}
(a) This is a direct consequence of Lemma \ref{lem:PS-min-proj}(b) and Proposition \ref{prop:atom-hered-q-distinct}(b).

\smnoind
(b) Notice that elements in $\CC^B$ are q-subsets because of parts (c) and (d) of Lemma \ref{lem:PS-min-proj}. 
It remains to verify Condition (S3). 
Suppose that $C_1,C_2\in \CC^B$ such that $C_1$ q-commutes with $C_2$. 
Lemma \ref{lem:PS-min-proj}(d) produces $p_1,p_2\in \CC(B)$ with $C_k=\KP^B[p_k]$ ($k=1,2$), and we know from Lemma \ref{lem:PS-min-proj}(e) that $p_1p_2 = p_2 p_1$. 
By Lemma \ref{lem:sum-orth-closed}, the projection $p_1 \vee p_2 = p_1 + p_2 - p_1p_2$ belongs to $\CC(B)$. 
Now, parts (c) and (d) of Lemma \ref{lem:PS-min-proj} implies that $C_1\vee C_2 = \KP^B[p_1\vee p_2]\in \CC^B$. 

\smnoind
(c) For an irreducible representation $(\pi, \KH)$ of $B$, we denote by $\ti \pi: B^{**}_\ba \to \CB(\KH)$ the weak$^*$-continuous extension of $\pi$. 
It is well-known that 
$$[\pi]_\sim \mapsto \ker \ti \pi$$ 
is a bijection from $\widehat{B}$ to the set of  maximal weak$^*$-closed ideals of $B^{**}_\ba$.
Moreover, maximal weak$^*$-closed ideals of $B^{**}_\ba$ are of the form $B^{**}_\ba (1-p)$ for a minimal central projection $p\in \CP_{B^{**}_\ba}$. 
On other hand, by parts (a), (c) and (e) of Lemma \ref{lem:PS-min-proj}, the set of minimal q-central q-subset in $\CQ(\KP^B)$ is precisely  
$$\big\{\KP^B[p]: p\in \CP_{B^{**}_\ba} \text{ is a minimal central projection}\big\}.$$ 
Consider a minimal central projection $p\in \CP_{B^{**}_\ba}$ and  $\phi \in \KP^B$.
Then $\bs_\phi \in \CP_{B^{**}_\ba}[p]$ (i.e, $\bs_\phi \leq p$) if and only if 
$$B^{**}_\ba(1-p)\subseteq \{x\in B^{**}_\ba: \phi(x^*x) =0 \},$$
which is equivalent to $B^{**}_\ba(1-p)\subseteq \ker \ti \pi_\phi$. 
However, as both $B^{**}_\ba(1-p)$ and $\ker \ti \pi_\phi$ are maximal weak$^*$-closed ideals of $B^{**}_\ba$, we know that they are the same. 
This gives the first statement. 

Let us consider $q\in \CC(B)$ to be a central projection, and let $L_q:= B^{**}_\ba(1-q)\cap  B$ be the corresponding closed ideal. 
As in the above, for any $\phi\in \KP^B$, one has $\phi\in \KP^B[q]$ if and only if $L_q \subseteq \ker \pi_\phi$.  
This, together with  Lemma \ref{lem:PS-min-proj}(d), gives the bijectivity between closed subsets of $\widehat{B}$ and q-central elements in $\CC^B$.

Finally, part (a) above and Corollary \ref{cor:q-cent-hered}(c) tell us that the q-complement of a q-central element in $\CQ(\KP^B)$ coincides with its ordinary complement and is also q-central. 
Thus, the statement concerning open subset of $\widehat{B}$ follows from the statement concerning closed subset of $\widehat{B}$. 
\end{proof}

\medskip

\begin{prop}
Let $A$ and $B$ be two $C^*$-algebras. 
If $\Theta: B\to A$ is a Jordan $^*$-isomorphism, then $\Theta^*|_{\KP^A}: \KP^A \to \KP^B$ is a strict quantum homeomorphism. 
\end{prop}

\medskip

Indeed, it is well-known that $\Theta^*(\KP^A) = \KP^B$ and $\Theta^*$ respects the q-distinctness relations. 
Moreover, as $\Theta^{**}(\CC(A)) = \CC(B)$ (note that a Jordan $^*$-isomorphism is an order isomorphism), we know from Lemma \ref{lem:PS-min-proj}(d) that  $\Theta^*|_{\KP^A}(\CC^A) = \CC^B$. 

\medskip

Now, we can present the main theorem of this section. 
Let us denote by $\widehat{A}^2$ the set of all irreducible $^*$-representations of a $C^*$-algebra $A$ with dimensions dominated  by $2$ (i.e. including those 1-dimensional ones). 
Set $A_0:=\bigcap_{\pi\in \widehat{A}^2}\ker \pi$, and consider $j_A: \KP^{A_0}\to \KP^A$ to be the map given by extensions of pure states. 

\medskip

\begin{thm}\label{thm:main2}
Let $A$ and $B$ be $C^*$-algebras, and $\Psi: \KP^A \to \KP^B$ be a strict quantum homeomorphism.

\smnoind
(a) If $\BM_2$ is not a quotient $C^*$-algebra of $A$, then there is a unique Jordan $^*$-isomorphism $\Gamma: B\to A$ such that $\Psi = \Gamma^*|_{\KP^A}$. 

\smnoind
(b) There is a Jordan $^*$-isomorphism $\Gamma_0: B_0\to A_0$ with $\Psi\circ j_A = j_B\circ \Gamma_0^*|_{\KP^{A_0}}$. 
\end{thm}
\begin{proof}
(a) For $p\in \CC(A)$, we know from the assumption on $\Psi$ and Lemma \ref{lem:PS-min-proj}(d) that there is a unique element $\Phi(p)\in \CC(B)$ satisfying
$$\KP^B\big[\Phi(p)\big]:= \Psi\big(\KP^A[p]\big).$$
The equality $\CC^B = \{\Psi(C): C\in \CC^A \}$ implies that  $\Phi$ is a bijection from $\CC(A)\setminus \{0\}$ onto $\CC(B)\setminus \{0\}$. 
Moreover, as $\Psi$ preserves the q-distinctness relations in both ways, we know that $\Phi$ is a strict quantum bijection. 
By Theorem \ref{thm:main}, there is a Jordan $^*$-isomorphism $\Theta:A\to B$ such that $\Theta^{**}(p) = \Phi(p)$ for every $p\in \CC(A)$. 
When $\psi\in \KP^B$, one has $\Theta^*(\psi)\in \KP^A$ and $\bs_{\Theta^*(\psi)} = \Phi^{-1}(\bs_\psi)$. 
Thus, if we set $\Gamma:= \Theta^{-1}$, then the required equality follows from Lemma \ref{lem:PS-min-proj}(b).

\smnoind
(b) Let $\KP^A_2:= \KG^{-1}(\widehat{A}^2)$ (see Proposition \ref{prop:quantum top}(c)), and set 
$$\KP^A_0 := \KP^A\setminus \KP^A_2.$$ 
By \cite[Proposition 4.4.10]{Ped79}, $\widehat{A}^2$ is a closed subset of $\widehat{A}$. 
Hence, Proposition \ref{prop:quantum top}(c) tells us that $\KP^A_2$ is a q-central element in $\CC^A$. 
Moreover, $\widehat{A}\setminus \widehat{A}^2$ is strictly quantum homeomorphic to $\widehat{A_0}$, under the canonical $^*$-homomorphism $\Delta_A: A \to M(A_0)$. 
In fact, $\omega \mapsto \omega\circ \Delta_A$ induces a bijection $\ti \Delta_A: \KP^{A_0} \to \KP^A_0$. 
Since $\ti \Delta_A$ preserves support projections (when we consider $A_0^{**}\subseteq A^{**}$ in the canonical way), it preserves the q-distinctness relations in both directions. 
Furthermore, as hereditary $C^*$-subalgebras of $A_0$ are precisely hereditary $C^*$-subalgebras of $A$ that are contained in $A_0$, we know that $\ti \Delta_A$ gives a bijection between quantum open subsets of $\KP^{A_0}$ and quantum open subsets of $\KP^A$ that are contained in quantum open subset $\KP^A_0$. 

On the other hand, by Proposition \ref{prop:quantum top}(c), if $\omega\in \KP^A$, then the set of elements whose GNS constructions equal to that of $\omega$ is the minimal q-central element in $\CQ(\KP^A)$ containing $\omega$. 
As $\Psi$ preserves the q-distinctness relations in both directions, it will preserve the minimal q-central q-subsets containing the corresponding elements in the respective Gelfand spectra. 
Consider $\omega\in \KP^A_2$. 
The dimension of the GNS construction of $\omega$ is either $1$ or $2$. 
If it is 1-dimensional, then $\{\omega \}$ is a q-central q-subset and so is $\{\Psi(\omega) \}$, which implies that the GNS construction of $\Psi(\omega)$ is 1-dimensional. 
Suppose that the GNS construction of $\omega$ is 2-dimensional.
Then there exists a unique element $\omega^\bot$ in the minimal q-central q-subset containing $\omega$ such that $\omega\neq_\mo\omega^\bot$. 
From this, one  can find a unique element $\Psi(\omega)^\bot$ in the minimal q-central q-subset containing $\Psi(\omega)$ with $\Psi(\omega)\neq_\mo \Psi(\omega)^\bot$. 
This implies that the GNS construction of $\Psi(\omega)$ is 2-dimensional.
Therefore, $\Psi(\KP^A_2) \subseteq \KP^B_2$. 
By symmetry, we know that $\Psi(\KP^A_2) = \KP^B_2$.

Finally, it is not hard to check, via the maps $\ti \Delta_A$ and $\ti\Delta_B$, that $\Psi$ induces a strict quantum homeomorphism $\bar\Psi: \KP^{A_0} \to \KP^{B_0}$. 
The conclusion then follows from part (a).  
\end{proof}



\medskip

\begin{cor}\label{cor:main2}
Let $A$ and $B$ be $C^*$-algebras. 
Suppose that  there is a strict quantum homeomorphism $\Psi: \KP^A \to \KP^B$.

\smnoind
(a) If $A$ is a primitive $C^*$-algebra such that $\BM_2$ is not a quotient $C^*$-algebra of $A$, then there is a map $\Theta: A\to B$ which is either a $^*$-isomorphism or a $^*$-anti-isomorphism such that $\Psi^{-1} = \Theta^*|_{\KP^B}$. 

\smnoind
(b) If $A = \BM_2$, then $A$ and $B$ are either $^*$-isomorphic or $^*$-anti-isomorphic. 
\end{cor}
\begin{proof}
(a) This follows from Theorem \ref{thm:main2}(a) and the well-known fact that Jordan $^*$-isomorphism between primitive $C^*$-algebras are either a $^*$-isomorphism or a $^*$-anti-isomorphism.

\smnoind
(b) Notice that 
\begin{equation}\label{eqt:q-closed-M2}
\CC^{\BM_2} = \{\emptyset\} \cup \big\{\KP^{\BM_2} \big\} \cup \big\{\{\omega\}: \omega\in \KP^{\BM_2} \big\}.
\end{equation}
For every $\omega \in \KP^{\BM_2}$, there exists exactly one $\omega^\bot\in \KP^{\BM_2}$ with  $\omega \neq_\mo \omega^\bot$. 
From this, and the existence of a strict quantum bijection from $\KP^{\BM_2}\setminus \{0\}$ to $\KP^B\setminus \{0\}$, we see that $B$ has exactly one irreducible $^*$-representation, and that this representation is  of dimension $2$ (see the proof of Theorem \ref{thm:main2}(b)). 
Consequently,  $B\cong \BM_2$. 
\end{proof}

\medskip

\begin{eg}\label{eg:M2}
	Fix an element $\omega_1\in \KP^{\BM_2}$ and consider $\omega_1^\bot$ to be an element in $\KP^{\BM_2}$ as in the proof of Corollary \ref{cor:main2}(b). 
	Define $\Psi:\KP^{M_2}\to \KP^{M_2}$ such that $\Psi(\omega_1) = \omega_1^\bot$, $\Psi(\omega_1^\bot) = \omega_1$ and $\Psi(\omega) = \omega$ when $\omega\notin \{\omega_1, \omega_1^\bot\}$. 
	Then $\Psi$ is a strict quantum homeomorphism (see \eqref{eqt:q-closed-M2}), but it cannot be induced by a Jordan $^*$-automorphism on  $\BM_2$ (because $\Psi$ cannot be extended to an affine map on the state space of $\BM_2$). 
\end{eg}

\medskip

\section{Conclusion}

\medskip

As seen in Proposition \ref{prop:atom-hered-q-distinct}, the naive notion of a quantum set (i.e., a set with a more restrictive form of distinctness between its elements) actually capture everything about quantum logic and can be regarded as the underlying ``quantum set theory'' behind quantum logic. 
Moreover, some  structures related to Physics are related to quantum sets (see Examples \ref{eg:uncert} and \ref{eg:tran-prob}). 

\medskip

On the other hand, the notion of quantum topological spaces can be regarded as semi-classical objects (as they are quantum sets equipped with special collections of subsets). 
A particular semi-classical object, known as the Gelfand spectrum, that associate with a quantum system modeled on the self-adjoint part $B_\sa$ of a $C^*$-algebra $B$, remembers $B_\sa$ up to Jordan isomorphism, when $B$ has no $2$-dimensional irreducible representation. 
In other words, one obtains a faithful image of quantum system in term of quantum topological space (under a mild assumption). 
Consequently, the semi-classical objects of Gelfand spectra serve a link from quantum systems to classical systems.    

\medskip

\section*{Acknowledgement}

\medskip

The authors are supported by the National Natural Science Foundation of China (11871285). 
The first named author is also supported by Nankai Zhide Foundation and the second named author is also supported by China Scholarship Council (201906200101). 

The two authors would also like to thank C.A. Akemann and L.G. Brown for some discussions on Proposition \ref{prop:atom-Thm2.2-APT}. 



\begin{thebibliography}{99}

\bibitem{Adams}
D.H. Adams, The completion by cuts of an orthocomplemented modular lattice, Bull. Austral. Math. Soc. \textbf{1} (1969), 279-280.




\bibitem{AHS}
E.M. Alfsen, H. Hanche-Olsen and F.W. Schultz, State spaces of $C^*$--algebras, Acta Math. \textbf{144} (1980), 267-305. 
 
\bibitem{Akemann68}
C.A. Akemann, Sequential convergence in the dual of a von Neumann algebra, Comm. Math. Phys. \textbf{7} (1968), 222--224.

\bibitem{Ake69}
C.A. Akemann, The General Stone-Weierstrass Problem, J. Funct. Anal. \textbf{4} (1969), 277--294.

\bibitem{Ake71}
C.A. Akemann, A Gelfand representation theory for $C^*$-algebras,
Pacific J. Math. \textbf{39} (1971), 1--11.

\bibitem{APT73}
C.A. Akemann, G.K. Pedersen and J. Tomiyama, Multipliers of $C^*$-algebras,
J.\ Funct.\ Anal., \textbf{13} (1973), 277-301.

\bibitem{Ban}
T. Banica, Quantum automorphism groups of homogeneous graphs, J. Funct. Anal. \textbf{224} (2005), 243-280. 

\bibitem{BBC}
T. Banica, J. Bichon and G. Chenevier, Graphs having no quantum symmetry, Ann. Inst. Fourier (Grenoble) \textbf{57} (2007), 955-971. 



\bibitem{BGS}
J. Bhowmick, D. Goswami and A. Skalski, Quantum isometry groups of 0-dimensional manifolds, Trans. Amer. Math. Soc. \textbf{363} (2011), 901-921.

\bibitem{Brown14}
L.G. Brown, Large $C^*$-algebras of universally measurable operators, Quart. J. Math. \textbf{65} (2014), 851-855. 


\bibitem{BW93}
L.J. Bunce and J.D.M. Wright, On Dye's theorem for Jordan operator algebras, Exposition. Math. \textbf{11} (1993), 91--95. 

\bibitem{DRZ}
C. Dietzel, W. Rump and X. Zhang, One-sided orthogonality, orthomodular spaces, quantum sets, and a class of Garside groups, J. Algebra \textbf{526} (2019), 51-80. 


\bibitem{Dye}
H.A. Dye, On the geometry of projections in certain operator algebras, Ann. Math. \textbf{61} (1955), 73--89.


\bibitem{FR}
U.A. First and T. R\"{u}d, On Uniform admissibility of unitary and smooth representations, Arch. Math. \textbf{112} (2019), 169-179. 

\bibitem{GK}
R. Giles and H. Kummer, A non-commutative generalization of topology, Indiana Univ. Math. J. \textbf{21} (1971/72), 91-102.


\bibitem{Gud}
S. Gudder, Measure and integration in quantum set theory, in  \emph{Current issues in quantum logic (Erice, 1979)},
Ettore Majorana Internat. Sci. Ser.: Phys. Sci. \textbf{8}, Plenum, New York-London (1981), 341-352. 

\bibitem{Ham81}
M. Hamana, Regular embeddings of $C^*$-algebras in monotone complete $C^*$-algebras, 
J. Math. Soc. Japan \textbf{33} (1981), 159-183. 


\bibitem{Ham}
J. Hamhalter, Dye's theorem and Gleason's theorem for $AW^*$-algebras, J. Math. Anal. Appl. \textbf{422} (2015), 1103--1115. 



\bibitem{JSW} 
L. Junk, S. Schmidt and M. Weber, Almost all trees have quantum symmetry, Arch. Math. \textbf{115} (2020), 367-378. 

\bibitem{Kal83}
G. Kalmbach, \emph{Orthomodular lattices}, Lond. Math. Soc. Mono. \textbf{18}, Academic Press (1983). 


\bibitem{Kal98}
G. Kalmbach, \emph{Quantum measures and spaces}, Math. and its Appl. \textbf{453}, Kluwer Academic Publishers  (1998).




\bibitem{MacN}
H.M. MacNeille, Partially ordered sets, Trans. Amer. Math. Soc. \textbf{42} (1937), 416-460.


\bibitem{Mihara}
T. Mihara, Characterisation of the Berkovich spectrum of the Banach algebra of bounded continuous functions, Doc. Math. \textbf{19} (2014), 769-799. 

\bibitem{Mur}
G.J. Murphy, \emph{$C^*$-algebras and operator theory}, Academic Press (1990).

\bibitem{MRV}
B. Musto, D. Reutter, D. Verdon, The Morita theory of quantum graph isomorphisms, Comm. Math. Phys. \textbf{365} (2019), 797-845. 


\bibitem{Ng-Cat-QS}
C.K. Ng, Categories of operator algebras and ortho-sets, in preparation. 

\bibitem{NW}
C.K. Ng and N.C. Wong, A Murray-von Neumann type classification of $C^*$-algebras, in \emph{Operator Semigroups Meet Complex Analysis, Harmonic Analysis and Mathematical Physics, Herrnhut, Germany (in honor of Prof. Charles Batty for his 60th birthday)}, Operator Theory: Advance and Applications, \textbf{250}, Springer Internat. Publ. (2015), 369--395.

\bibitem{Ped72}
G.K. Pedersen, Applications of weak$^*$-semicontinuity in $C^*$-algebra theory, Duke Math. J. \textbf{39} (1972), 431-450.

\bibitem{Ped79}
G.K. Pedersen,
\emph{$C^*$-algebras and their automorphism groups},
Academic Press (1979).



\bibitem{Schles}
K.-G. Schlesinger, Toward quantum mathematics I: From quantum set theory to universal quantum mechanics, J. Math. Phys. \textbf{40} (1999), 1344-1358.

\bibitem{Sch}
S. Schmidt, The Petersen graph has no quantum symmetry, Bull. Lond. Math. Soc. \textbf{50} (2018), 395-400. 


\bibitem{Shu82}
F.W. Shultz, Pure states as a dual object for $C^*$-algebras, Comm. Math. Phys. \textbf{82} (1981/82), 497-509. 


\bibitem{SV07}
I. Stubbe and B. Van Steirteghem, Propositional systems, Hilbert lattices and generalized Hilbert spaces, in \emph{Handbook of quantum logic and quantum structures}, Elsevier Sci. B. V. (2007), 477-523. 





\bibitem{Takeuti}
G. Takeuti, Quantum set theory, in \emph{Current issues in quantum logic (Erice, 1979)}, 
Ettore Majorana Internat. Sci. Ser.: Phys. Sci., \textbf{8}, Plenum, New York-London (1981), 303-322.

\bibitem{TK}
S. Titani and H. Kozawa, Quantum set theory, Internat. J. Theoret. Phys. \textbf{42} (2003), 2575-2602.

\bibitem{Vlad}
D.A. Vladimirov, \emph{Boolean algebras in analysis}, Kluwer Academic (2002).

\bibitem{Walker}
J. W. Walker, From graphs to ortholattices and equivariant maps, 
J. Combin. Theory Ser. B \textbf{35} (1983), 171-192.

\end{thebibliography}
\end{document}